\documentclass[11pt,thmsa]{article}

\usepackage{amssymb,latexsym}
\usepackage{amsmath,amscd}
\usepackage[all]{xy}

 \usepackage[latin1]{inputenc}
 \usepackage{mathrsfs}
 \usepackage{dsfont}
 \usepackage{amsmath}
 \usepackage{amsthm}
 \usepackage{amsfonts}
 \usepackage{amssymb}
 \usepackage[dvips]{graphicx}
 \usepackage{wrapfig}
 \usepackage{layout}
 \usepackage{verbatim}
 \usepackage{alltt}
 \usepackage{xypic}
 \usepackage{bbm}
 \input xy
 \xyoption{all}

 \DeclareMathAlphabet{\mathpzc}{OT1}{pzc}{m}{it}

 \newtheorem{theorem}{Theorem}[section]
 \newtheorem{lemma}[theorem]{Lemma}

 \newtheorem{definition}[theorem]{Definition}
  \theoremstyle{definition}
 \newtheorem{example}[theorem]{Example}
\newtheorem{examples}[theorem]{Examples}
 \newtheorem{remark}[theorem]{Remark}

\newtheorem*{acknowledgements}{Acknowledgements}
\renewenvironment{proof}{\noindent{\it
Proof.}}{\bgroup\hspace{\stretch{1}}$\square$\egroup\medskip\par}

\newcommand{\smooth}{\mathcal{C}^{\infty}}
\newcommand{\sM}{\mathcal{M}}
\newcommand{\sN}{\mathcal{N}}
\newcommand{\Mor}{\textrm{Mor}}
\newcommand{\Map}{\textrm{Map}}
\newcommand{\Diff}{\textrm{Diff}}

\begin{document}
\vspace{15cm}
 \title{Introduction to supergeometry}
\author{Alberto S. Cattaneo\footnote{Institut f\"ur Mathematik, Universit\"at Z\"urich, alberto.cattaneo@math.uzh.ch. Partially supported by SNF Grants 20-113439 and 20-131813.} \hspace{0cm} and
Florian Sch\"atz\footnote{Center for Mathematical Analysis, Geometry and Dynamical Systems, IST Lisbon, fschaetz@math.ist.utl.pt. Partially supported by a FCT grant and project PTDC/MAT/098936/2008.}}

 \maketitle

\begin{abstract} 
These notes are based on a series of lectures given by the first author at the school of `Poisson 2010', held at IMPA, Rio de Janeiro. They contain an exposition
of the theory of super- and graded manifolds, cohomological vector fields, graded symplectic structures, reduction and the AKSZ-formalism.
\end{abstract}

\section{Introduction}\label{s:intro}

The main idea of supergeometry is to extend classical geometry by allowing for {\em odd} coordinates. These
are coordinates which anticommute, in contrast to usual coordinates which commute. The global object, which one obtains from gluing
such extended coordinate systems, are {\em supermanifolds}.
A prominent example is obtained by considering the one-forms
$(dx^i)_{i=1}^{n}$ as odd coordinates, accompanying the usual `even' coordinates $(x^i)_{i=1}^n$ on $\mathbb{R}^n$.
The corresponding supermanifold is known as $\Pi T\mathbb{R}^n$.

The use of odd coordinates has its roots in physics, but it turned out to have interesting mathematical applications as well.
Let us briefly mention those which are explained in more detail below:

\begin{itemize}
 \item[i)] Some classical geometric structures
can be encoded in simple supergeometric structures. For instance, Poisson manifolds and Courant algebroids
can be described in a uniform way in terms of supermanifolds equipped with a supersymplectic structure and a symplectic cohomological 
vector field, see Subsection \ref{subsection:examples}.
One can also treat generalized complex structures in this setting.

 \item[ii)] As an application, a unifying approach to the {\em reduction} of Poisson manifolds, Courant algebras and generalized complex
structures can be developed. An outline of this approach is presented in Subsection \ref{subsection:reduction}.

 \item[iii)] The AKSZ-formalism (\cite{AKSZ}) allows one to associate topological fields theories to supermanifolds equipped
with additional structures, see Section \ref{section:AKSZ}. Combinging this with the supergeometric description of classical geometric structures mentioned in $i)$,
one obtains topological field theories associated to Poisson manifolds and Courant algebroids. These field theories
include the Poisson-Sigma model as well as Chern-Simons theory for trivial principal bundles.
\end{itemize}

\subsection*{Plan of the notes}

Section \ref{section:supermanifolds} is a short review of the basics of supergeometry.

In Section \ref{section:graded_manifolds}, graded manifolds, as well as graded and cohomological vector fields,
are introduced. The concept
of cohomological vector fields allows one to think of `symmetries', which appear in a wide variety of examples,
in a unified and geometric way.
For instance, $L_{\infty}$-algebras and Lie algebroid structures can be seen as special instances
of cohomological vector fields.

In Section \ref{section:graded_symplectic_geometry}, graded symplectic manifolds
are explained. Due to the additional grading, graded symplectic geometry often behaves
much more rigidly than its ungraded counterpart.
A dg symplectic manifold is a graded symplectic manifold with a compatible cohomological vector field.
Poisson manifolds, Courant algebroids and generalized complex structures
fit naturally into the framework of dg symplectic manifolds. This is explained in Subsection \ref{subsection:examples},
while Subsection \ref{subsection:reduction} outlines a unified approach to the reduction of these structures via
graded symplectic geometry.

Finally, Section \ref{section:AKSZ} provides an introduction to the AKSZ-formalism (\cite{AKSZ}).
This is a procedure which allows one to associate topological field theories to dg symplectic manifolds.
As particular examples, one recovers the Poisson Sigma model and Chern-Simons theory (for trivial principal bundles).
In this Section, we basically follow the expositions of the AKSZ-formalism from \cite{R2} and \cite{CF}, respectively.

\begin{acknowledgements}
We thank Dmitry Royenberg, Pavol \v{S}evera and Marco Zambon for helpful comments.
Moreover, we thank the school of `Poisson 2010' for partial financial support.
\end{acknowledgements}

\section{Supermanifolds}\label{section:supermanifolds}

\subsection{Definition}
A {\em supermanifold} $\sM$ is a locally ringed space $(M,\mathcal{O}_M)$ which is locally isomorphic to
\begin{align*}
(U, \smooth(U)\otimes \wedge W^*),
\end{align*}
where $U$ is an open subset of  $\mathbb{R}^n$ and $W$ is some finite-dimensional real vector space.
The isomorphism mentioned above is in the category of $\mathbb{Z}_2$-graded algebras, i.e. the {\em parity}
\begin{align*}
\bigoplus_{k\ge 0}\smooth(U)\otimes \wedge^k W^* \to \mathbb{Z}_2, \quad f\otimes x \mapsto |f\otimes x|:= |x| = k \textrm{ mod } 2
\end{align*}
has to be preserved.

Loosly speaking, every supermanifold is glued from pieces that look like open subsets of $\mathbb{R}^n$, together
with some odd coordinates, which correspond to a basis of $W^*$. 
The supermanifold corresponding to such a local
piece is denoted by $U \times \Pi W$, and we write 
\begin{align*}
\smooth(U\times \Pi W):= \smooth(U) \otimes \wedge W^*.
\end{align*}
Similarly, the algebra of {\em polynomial functions} on $V\times \Pi W$, for $V$ and $W$ real, finite dimensional vector spaces,
is $S(V^*)\otimes \wedge W^*$. Here $S(V^*)$
denotes the symmetric algebra of the vector space $V^*$.

In the global situation, the algebra of smooth function $\smooth(\sM)$ on a supermanifold $\sM$
is defined to be the algebra of global sections of the sheaf associated to $\sM$. The parity extends to
$\smooth(\sM)$ and $\smooth(\sM)$ is a graded commutative algebra with respect to this parity, i.e.
for $f$ and $g$ homogeneous elements of degree $|f|$ and $|g|$ respectively, one has
\begin{align*}
 f\cdot g = (-1)^{|f||g|}g\cdot f.
\end{align*}

\begin{examples}
\hspace{0cm}
\begin{enumerate}
 \item The algebra of differential forms $\Omega(M)$ on a manifold $M$ is locally isomorphic to
\begin{align*}
 \smooth(U) \otimes \wedge T_x^*M
\end{align*}
where $x$ is some point on $U$. Hence the sheaf of differential forms on a manifold corresponds to a supermanifold.

\item Let $\mathfrak{g}$ be a real, finite dimensional Lie algebra. The cochains of the Chevalley-Eilenberg complex of $\mathfrak{g}$
are the elements of $\wedge \mathfrak{g}^*$, which is the same as the algebra of smooth functions on the supermanifold 
$\Pi \mathfrak{g}$.
\end{enumerate}
\end{examples}

\subsection{Morphisms of supermanifolds}
Since supermanifolds are defined as certain locally ringed spaces, it is natural to define morphisms of supermanifolds as morphims
of these locally ringed spaces. In the smooth setting, one can equivalently define morphisms from $\sM$ to $\sN$
to be morphisms of superalgebras from $\smooth(\sN)$ to $\smooth(\sM)$, see \cite{Var} or \cite{CFio}.

Let us spell this out in more detail for the local case, i.e. 
consider a patch of $\sN$ isomorphic to $\tilde{V}\times \Pi \tilde{W}$.
In this situation one has

\begin{align*}
\textrm{Mor}(\sM,\sN) \cong  \textrm{Mor}_{\textrm{alg}}(\smooth(\sN),\smooth(\sM)) \cong 
\left((\tilde{V}\oplus \Pi\tilde{W})\otimes \smooth(\sM) \right)_{\textrm{even}}.
\end{align*}
Here, the last object is a super vector space, i.e. a vector space with a decomposition
into an even and an odd part. $\smooth(\sM)$ is graded by the parity, $\tilde{V}\oplus \Pi\tilde{W}$ is considered with its obvious
decomposition and the parity of a tensor product $a\otimes b$ is the product of the parities of $a$ and $b$, respectively. 

The last isomorphism uses the fact that it suffices to know the restriction of a morphism
\begin{align*}
\smooth(\tilde{V}\times \Pi \tilde{W}) \to \smooth(\sM)
\end{align*}
to the subspace of linear functions,
in order to be able to recover the morphism itself.

\section{Graded manifolds}\label{section:graded_manifolds}

\subsection{Definition}
Let us first introduce the relevant linear theory:
A {\em graded vector space} $V$ is a collection of vector spaces $(V_i)_{i\in \mathbb{Z}}$.
The algebra of {\em polynomial functions} on $V$ is the {\em graded symmetric algebra} $S(V^*)$ over $V^*$.
In more detail:
\begin{itemize}
\item The {\em dual} $V^*$ of a graded vector space $V$ is the graded vector space $(V_{-i}^*)_{i\in \mathbb{Z}}$.
\item The {\em graded symmetric algebra} $S(W)$ over a graded vector space $W$
is the quotient
of the tensor algebra of $W$ by the ideal generated by the elements of the form
\begin{align*}
 v\otimes w -(-1)^{|v||w|}w\otimes v
\end{align*}
for any homoegenous elements $v$ and $w$ of $W$.
\end{itemize}

A {\em morphism} $f: V \to W$ of graded vector spaces is a collection of linear maps
\begin{align*}
(f_i: V_i \to W_i)_{i\in \mathbb{Z}}.
\end{align*} 
The morphisms between graded vector spaces are also referred to as graded linear maps.

Moreover, $V$ {\em shifted} by $k$ is the graded vector space $V[k]$ given by $(V_{i+k})_{i\in \mathbb{Z}}$.
By definition, a graded linear map of degree $k$ between $V$ and $W$ is
a graded linear map between $V$ and $W[k]$.

The definition of a {\em graded manifold} is 
analogous to that of  supermanifold, but now in the graded setting:
\begin{itemize} 
\item The local model is 
\begin{align*}
(U,\smooth(U)\otimes  S(W^*)),
\end{align*}
where $U$ is an open subset of $\mathbb{R}^{n}$ and $W$ is a graded vector space.
\item The isomorphism between the structure sheaf and the local model is in the category of $\mathbb{Z}$-graded algebras.
\end{itemize}
The algebra of smooth functions of a graded manifold $(M,\mathcal{O}_M)$ (i.e. algebra of global sections) automatically inherits a $\mathbb{Z}$-grading. 
Morphisms between graded manifolds are morphisms of locally ringed spaces. In the smooth setting, one can equivalently
consider morphisms of the $\mathbb{Z}$-graded algebra of smooth functions.

Essentially all examples of graded manifolds come from {\em graded vector bundles}, which are generalizations of graded vector spaces.
A {\em graded vector bundle} $E$ over a manifold $M$ is a collection
of ordinary vector bundles $(E_i)_{i\in \mathbb{Z}}$ over $M$.
The sheaf 
\begin{align*}
U \mapsto \Gamma(U,S(E|_{U}^*))
\end{align*}
corresponds to a graded manifold which we will also denote by $E$ from now on. It can be shown that any graded manifold is isomorphic to a graded manifold
associated to a graded vector bundle.

\begin{examples}
\hspace{0cm}
\begin{enumerate}
 \item For $V$ an ordinary vector space, one has $\smooth(V[1])=\wedge V^*$. In particular,
the space of cochains of the Chevalley-Eilenberg complex of a Lie algebra $\mathfrak{g}$ is equal to
$\smooth(\mathfrak{g}[1])$.
 \item The algebra of differential forms $\Omega(M)$ is the algebra of smooth functions on $T[1]M$.
\end{enumerate}
\end{examples}

\subsection{Graded vector fields}

Let $V$ be a graded vector space with homogeneous coordinates $(x^i)_{i=1}^{n}$ corresponding to a basis of $V^*$. A {\em vector field} on $V$
is a linear combination of the form
\begin{align*}
X=\sum_{i=1}^{n}X^{i}\frac{\partial}{\partial x^{i}}
\end{align*}
where $(X^{i})_{i=1}^{n}$ is a tuple of functions on $V$, i.e. of elements of $S(V^{*})$,
and $(\frac{\partial}{\partial x^{i}})_{i=1}^{n}$ is the basis of $V$ dual to $(x^{i})_{i=1}^{n}$.
The vector field $X$ acts on the algebra of functions according to the following rules:
\begin{itemize}
\item $\frac{\partial}{\partial x^{i}}(x^{j}) = \delta^{j}_{i}$ and
\item $\frac{\partial}{\partial x^{i}}(fg)= \left(\frac{\partial}{\partial x^{i}}(f) \right) g + (-1)^{|x^{i}||f|}f \left(\frac{\partial}{\partial x^{i}}(g) \right).$
\end{itemize}
A vector field is {\em graded} if it maps functions of degree $m$ to functions of degree $m+k$ for some fixed $k$. In this case, the integer $k$ is called the degree of $X$.

Globally, graded vector fields on a graded manifold can be identified with graded derivations of the algebra of smooth functions. Accordingly, a graded vector field on $\sM$ is a graded linear map
\begin{align*}
X: \smooth(\sM) \to \smooth(\sM)[k]
\end{align*}
which
satisfies the graded Leibniz rule, i.e.
\begin{align*}
X(fg)=X(f)g + (-1)^{k|f|} fX(g)
\end{align*}
holds for all homogenoeus smooth functions $f$ and $g$.
The integer $k$ is called the {\em degree} of $X$.

\begin{example}
Every graded manifold comes equipped with the {\em graded Euler vector field} which can be
defined in two ways:
\begin{itemize}
\item In local coordinates $(x^{i})_{i=1}^{n}$, it is given by 
\begin{align*}
E=\sum_{i=1}^{n}|x^{i}|x^i\frac{\partial}{\partial x^{i}}.
\end{align*}
\item Equivalently, it is the derivation which acts on homogeneous smooth functions via
\begin{align*}
E(f) = |f| f.
\end{align*}
\end{itemize}
\end{example}

\subsection{Cohomological vector fields}

\begin{definition}
A {\em cohomological vector field} is a graded vector field of degree $+1$
which commutes with itself.
\end{definition}

\begin{remark}
Observe that the graded commutator equips the graded vector space of graded vector fields
with the structure of a graded Lie algebra: if $X$ and $Y$ are graded derivations of degree $k$ and $l$ respectively, then
\begin{align*}
[X,Y]:= X \circ Y - (-1)^{kl} Y \circ X
\end{align*}
is a graded derivation of degree $k+l$.

It is left as an exercise to the reader to verify that in local coordinates $(x_i)_{i=1}^{n}$, the graded
commutator $[X,Y]$ of
\begin{align*}
X=\sum_{i=1}^{n}X^{i}\frac{\partial}{\partial x^{i}} \quad \textrm{and} \quad Y=\sum_{j=1}^{n}Y^{j}\frac{\partial}{\partial x^{j}}
\end{align*}
is equal to
\begin{align*}
\sum_{i,j=1}^{n}X^{i}\frac{\partial Y^{j}}{\partial x_{i}} \frac{\partial}{\partial x^{j}} - (-1)^{\sharp}\sum_{i,j=1}^{n}Y^{j}\frac{\partial X^{i}}{\partial x_{j}} \frac{\partial}{\partial x^{i}},
\end{align*}
for some appropriate sign $\sharp$.

Futhermore, one can check that the degree of a graded  vector field $X$ is its eigenvalue with respect to the Lie derivative along $E$, i.e.
\begin{align*}
[E,X]=\deg(X)X.
\end{align*}

Let $Q$ be a graded vector field of degree $+1$, i.e. $Q$ is a linear map
\begin{align*}
Q: \smooth(\sM) \to \smooth(\sM)[1]
\end{align*}
which satisfies the graded Leibniz rule.
Because of
\begin{align*}
[Q,Q] = 2 \left(Q \circ Q\right),
\end{align*}
every cohomological vector field on $\sM$ corresponds to a differential on the graded algebra of smooth functions $\smooth(\sM)$.
\end{remark}

\begin{examples}
\hspace{0cm}
\begin{enumerate}
\item Consider the shifted tangent bundle $T[1]M$, whose algebra of smooth functions is equal to the algebra
of differential forms $\Omega(M)$.
The de Rham differential on $\Omega(M)$ corresponds to a cohomological vector field $Q$ on $T[1]M$.

We fix local coordinates $(x^i)_{i=1}^{n}$ on $M$ and denote the induced fiber coordinates on $T[1]M$ by $(dx^{i})_{i=1}^{n}$.
In the coordinate system $(x^i,dx^i)_{i=1}^{n}$, the cohomological vector field $Q$ is given by
\begin{align*}
Q=\sum_{i=1}^{n}dx^{i}\frac{\partial}{\partial x^{i}}.
\end{align*}
\item Let $\mathfrak{g}$ be a real, finite dimensional Lie algebra.
The graded manifold $\mathfrak{g}[1]$ carries a cohomological vector field $Q$ which corresponds to the
Chevalley-Eilenberg differential on $\wedge \mathfrak{g}^{*} = \smooth(\mathfrak{g}[1])$.

In more detail, let $(e_i)_{i=1}^{n}$ be a basis of $\mathfrak{g}$, and $(f_{ij}^{k})$ be the corresponding structure constants given by
\begin{align} \label{structure_constants}
[e_i,e_j]=\sum_{k=1}^{n}f_{ij}^{k}e_k.
\end{align}
Then, the cohomological vector field $Q$ reads
\begin{align*}
\frac{1}{2}\sum_{i,j,k=1}^{n}x^{i}x^{j}f_{ij}^{k}\frac{\partial}{\partial x^{k}},
\end{align*}
where $(x^{i})_{i=1}^{n}$ are the coordinates on $\mathfrak{g}[1]$ which correspond to the basis dual to $(e_i)_{i=1}^{n}$.

Actually, one can check that $[Q,Q]=0$ is equivalent to the statement that the bracket $[-,-]: \mathfrak{g}\otimes \mathfrak{g} \to \mathfrak{g}$ defined via formula (\ref{structure_constants})
satisfies the Jacobi-identity.

\item One can generalize the last example in two directions:
\begin{itemize}
\item allowing for {\em higher degrees}: let $V$ be a graded vector space with only finitely 
many non-zero homogeneous components, all of which are finite dimensional. Formal\footnote{{\em Formal} cohomological vector fields
are elements of the completion of the space of vector fields with respect to the degree, i.e. one considers $\hat{S}(V^*)\otimes V$
instead of $S(V^*)\otimes V$. The subset of (ordinary) cohomological vector fields corresponds to $L_{\infty}$-algebra structures on $V$ whose structure maps
$V^{\otimes n} \to V[1]$ vanish for all but finitely many $n$.} cohomological vector fields on $V$ are in one-to-one
correspondence with $L_{\infty}$-algebra structures on $V$.
\item allowing for a {\em non-trivial base}: let $A$ be a vector bundle over a manifold $M$.
Cohomological vector fields on $A[1]$ are in one-to-one correspondence with
Lie algebroids structures on $A$. This observation is due to Vaintrob \cite{V}.
\end{itemize}
\end{enumerate}
\end{examples}

\begin{definition}
A graded manifold endowed with a cohomological vector field is called a {\em differential graded manifold}, or dg manifold for short.

A morphism of dg manifolds is a morphism of graded manifolds, with respect to which the cohomological vector field are related.
\end{definition}

\begin{remark}
Morphisms of dg manifolds can be defined equivalently by requiring that the corresponding
morphisms between the algebra of smooth functions is a chain map with respect to the differentials
given by the cohomological vector fields.
\end{remark}

\begin{remark}
So far we elaborated on the condition $[Q,Q]=0$ mostly from an algebraic perspective.
However, it has also geometric significance, as we will see now.
Assume $X$ is a graded vector field of degree $k$ on a graded manifold $\sM$. We want to
construct the flow of $X$, i.e. solve the ordinary differential equation
\begin{align} \label{ode}
\frac{dx^{i}(t)}{dt}=X^{i}(x(t))
\end{align}
in a coordinate chart of $\sM$, where $X$ is given by
\begin{align*}
\sum_{i=1}^{n}X^{i}\frac{\partial}{\partial x^{i}}.
\end{align*}
Observe that the degree of the components $X^{i}$ is $|x^{i}|+k$.
Consequently we have to assign degree $-k$ to the time parameter $t$ in order for the two sides
of the flow equation to have the same degree. Although we will not introduce the concept of maps between
graded manifolds until Section \ref{section:AKSZ}, let us mention that one can think of the solution of equation (\ref{ode}) as a map
\begin{align*}
\mathbb{R}[k] \to \sM.
\end{align*}
Now, assume that $X$ is of degree $+1$. This implies that $t$ is of degree $-1$ and hence
squares to zero. The expansion of the flow with respect to $t$ looks like
\begin{align*}
x^{i}(t) = x^{i} + t v^{i}
\end{align*}
where $v^{i}$ is of degree $|x^{i}|+1$.
On the one hand, this implies
\begin{align*}
\frac{d x^{i}(t)}{dt} = v^{i}
\end{align*}
while on the other hand,
\begin{align*}
X^{i}(x(t))=X^{i}(x) + t \sum_{j=1}^{n} \frac{\partial X^{i}}{\partial x^j}v^{j}
\end{align*}
holds.
Using the flow equation, we obtain
\begin{align*}
v^{i}=X^{i}(x) \qquad \textrm{and} \qquad \sum_{j=1}^{n}\frac{\partial X^{i}}{\partial x^j}v^{j}=0,
\end{align*}
which combines into
\begin{align*}
\sum_{j=1}^{n}\frac{\partial X^{i}}{\partial x^j}X^{j} = 0 \quad \Leftrightarrow \quad [X,X]=0.
\end{align*}
So $[X,X]=0$ turns out to be a necessary and sufficient condition for the integrability of $X$.
This conclusion can be seen as a special instance of Frobenius Theorem for smooth graded manifolds, see \cite{BCMZ2}.
\end{remark}

\section{Graded symplectic geometry}\label{section:graded_symplectic_geometry}

\subsection{Differential forms}

Locally, the algebra of differential forms on a graded manifold $\sM$ is constructed by adding new coordinates $(dx^{i})_{i=1}^{n}$
to a system of homogeneous local coordinates $(x^{i})_{i=1}^{n}$ of $\sM$. Moreover, one assigns
the degree $|x^{i}|+1$ to $dx^{i}$.

\begin{remark}
\hspace{0cm}
\begin{enumerate}
\item If $x^{i}$ is a coordinate of odd degree, $dx^i$ is of even degree and consequently
\begin{align*}
(dx^{i})^{2}\neq 0.
\end{align*}
\item On an ordinary smooth manifold, differential forms have two important properties:
they can be differentiated -- hence the name differential forms -- and they also provide
the right objects for an integration theory on submanifolds. It turns out that on graded manifolds, this is no longer true,
since the differential forms we introduced do not come along with a nice integration theory.
To solve this problem, one has to introduce new objects, called `integral forms'.
The interested reader is referred to \cite{M}.
\end{enumerate}
\end{remark}

A global description of differential forms on $\sM$ is as follows:
the shifted tangent bundle $T[1]\sM$ carries a natural structures of a dg manifold with
a cohomological vector field $Q$ that reads
\begin{align*}
\sum_{i=1}^{n}dx^{i}\frac{\partial}{\partial x^{i}}
\end{align*}
in local coordinates.
The {\em de Rham complex} $(\Omega(\sM),d)$ of $\sM$ is $\smooth(T[1]\sM)$, equipped
with the differential corresponding to the cohomological vector field $Q$.

\begin{remark}
The classical Cartan calculus extens to the graded setting:
\begin{enumerate}
\item Let $X$ be a graded vector field on $\sM$ and $\omega$ a differential form.
Assume their local expressions in a coordinate system $(x^{i})_{i=1}^{n}$ are
\begin{align*}
X=\sum_{i=1}^{n}X^i\frac{\partial}{\partial x^{i}} \qquad \textrm{and} \qquad \omega=\sum_{i_1,\cdots,i_k=1}^{n}\omega_{i_1\cdots i_k}dx^{i_1}\cdots dx^{i_k}, \qquad \textrm{respectively.}
\end{align*}
The {\em contraction} $\iota_{X}\omega$ of $X$ and $\omega$ is given locally by
\begin{align*}
\sum_{i_{1},\cdots,i_k=1}^{n}\sum_{l=1}^{k}\pm \omega_{i_{1}\cdots i_k}X^{i_l}dx^{i_1}\cdots \widehat{dx^{i_l}}\cdots dx^{i_k}.
\end{align*}
Alternatively -- and to get the signs right -- one uses the following rules:
\begin{itemize}
\item $\iota_{\frac{\partial}{\partial x^{i}}}dx^{j}=\delta_i^j$,
\item $\iota_{\frac{\partial}{\partial x^{i}}}$ is a graded derivation of the algebra of differential forms of degree $|x^{i}|-1$.
\end{itemize}
\item The {\em Lie derivative} is defined via Cartan's magic formula
\begin{align*}
L_{X}\omega:= \iota_{X}d\omega + (-1)^{|X|} d\iota_{X} \omega.
\end{align*}
If one considers both $\iota_X$ and $d$ as graded derivations, $L_{X}$ is their graded commutator $[d,\iota_X]$.
It follows immediately that
\begin{align*}
[L_X,d]=0
\end{align*}
holds.
\end{enumerate}

Let us compute the Lie derivative with respect to the graded Euler vector field $E$ in local coordinates
$(x^{i})$. By definition
\begin{align*}
L_Ex^{i} = |x^{i}|x^{i}
\end{align*}
and consequently
\begin{align*}
L_{E}dx^{i}=dL_{E}x^{i}=|x^{i}|dx^{i}.
\end{align*}
This implies that $L_E$ acts on homogenous differential forms on $\sM$ by multiplication by
the difference between the {\em total degree} -- i.e. the degree in $\smooth(T[1]\sM)$ -- and the {\em form degree} --
which is given by counting the `$d$'s, loosly speaking. We call this difference the {\em degree} of a differential form $\omega$ and denote it by $\deg{\omega}$. 
\end{remark}

\subsection{Basic graded symplectic geometry}

\begin{definition}
A graded symplectic form of degree $k$ on a graded manifold $\sM$ is a two-form $\omega$ which has the following properties:
\begin{itemize}
\item $\omega$ is homogeneous of degree $k$,
\item $\omega$ is closed with respect to the de Rham differential,
\item $\omega$ is non-degenerate, i.e. the induced morphism of graded vector bundles
\begin{align*}
\omega: T\sM \to T^{*}[k]\sM
\end{align*}
is an isomorphism.
\end{itemize}

A graded symplectic manifold of degree $k$ is a pair $(\sM,\omega)$ of a graded manifold $\sM$ and a graded symplectic form $\omega$ of degree $k$ on $\sM$.
\end{definition}

\begin{examples}
\hspace{0cm}
\begin{enumerate}
\item Ordinary symplectic structures on smooth manifolds can be seen as graded symplectic structures.
\item Let $V$ be a real vector space. The contraction between $V$ and $V^{*}$ defines a symmetric non-degenerate pairing on $V\oplus V^{*}$.
This pairing is equivalent to a symplectic form of degree $k+l$ on $V[k]\oplus V^{*}[l]$.
\item Consider $\mathbb{R}[1]$ with the two-form $\omega=dx dx$. This is a symplectic form of
degree $2$.
\end{enumerate}
\end{examples}

\begin{lemma}\label{lemma:Dmitry1}
Let $\omega$ be a graded symplectic form of degree $k\neq 0$.
Then $\omega$ is exact.
\end{lemma}

\begin{proof}
One computes
\begin{align*}
k\omega = L_E \omega = d \iota_{E}\omega.
\end{align*}
This implies $\omega = \frac{d\iota_{E}\omega}{k}$.
\end{proof}

\begin{definition}
Let $\omega$ be a graded symplectic form on a graded manifold $\sM$.

A vector field $X$ is called ...
\begin{itemize}
\item {\em symplectic} if the Lie derivative of $\omega$ with respect to $X$ vanishes, i.e. $L_X\omega=0$,
\item {\em Hamiltonian} if the contraction of $X$ and $\omega$ is an exact one-form, i.e.
there is a smooth function $H$ such that $\iota_X\omega = dH$.
\end{itemize}
\end{definition}

\begin{lemma}\label{lemma:Dmitry2}
Suppose $\omega$ is a graded symplectic form of degree $k$ and $X$
is a symplectic vector field of degree $l$.
If $k+l\neq 0$, then $X$ is Hamiltonian.
\end{lemma}

\begin{proof}
By definition, we have
\begin{align*}
[E,X]= l X \qquad \textrm{and} \qquad L_X\omega = d\iota_X\omega = 0.
\end{align*}
Set $H:=\iota_{E}\iota_{X}\omega$ and compute
\begin{align*}
dH = d\iota_{E}\iota_{X}\omega = L_{E}\iota_X \omega - \iota_{E}d\iota_{X}\omega = \iota_{[E,X]} \omega = (k+l)\iota_{X}\omega.
\end{align*}
Hence $\iota_{X}\omega = \frac{dH}{k+l}$.
\end{proof}

\begin{remark}
Lemma \ref{lemma:Dmitry1} and Lemma \ref{lemma:Dmitry2} can be found in \cite{R1}.
\end{remark}

\begin{example}
Let $(\sM,\omega)$ be a graded symplectic manifold and $Q$ a symplectic cohomological vector field.
By definition, the degree of $Q$ is $1$. Lemma \ref{lemma:Dmitry2} implies that $Q$ is Hamiltonian
if $k= \deg\omega \neq -1$. We remark that the exceptional case $k = -1$ is also relevant since it appears in the BV-formalism (\cite{BV2,Sw,Se}).

Assume that $Q$ is Hamiltonian. Similar to the ungraded case, the graded symplectic form induces a
bracket $\{-,-\}$ via
\begin{align*}
\{f,g\}:=(-1)^{|f|+1}X_{f}(g)
\end{align*}
where $X_f$ is the unique graded vector field that satisfies $\iota_{X_f}\omega = df$.
It can be checked that $\{-,-\}$ satisfies relations similar to the ordinary Poisson bracket.
Using the bracket, one can express $Q$ with the help of a Hamiltonian function $S$ by
\begin{align*}
Q=\{S,-\}.
\end{align*} 
Since
\begin{align*}
[Q,Q](f) = \{\{S,S\},f\}
\end{align*}
The relation $[Q,Q]=0$ is equivalent to $\{S,S\}$ being a constant.

Observe that $S$ can be chosen of degree $k+1$, while the bracket $\{-,-\}$ decreases the degree by
$k$. Consequently, the degree of $\{S,S\}$ is $k+2$. Since constants are of degree $0$,
$k \neq -2$ implies
\begin{align*}
\{S,S\}=0.
\end{align*}
This last equation is known as the {\em classical master equation}.
\end{example}

\begin{definition}
A graded manifold endowed with a graded symplectic form and a symplectic cohomological
vector field is called a {\em differential graded symplectic manifold}, or dg symplectic manifold for short.
\end{definition}

\subsection{Examples of dg symplectic manifolds}\label{subsection:examples}

Next, we study some special cases of dg symplectic manifolds $(\sM,\omega,Q)$, where we
assume the cohomological vector field $Q$ to be Hamiltonian with Hamiltonian function $S$. As before, $k$ denotes the degree of
the graded symplectic form $\omega$.
As mentioned before, the case $k=-1$ occures in the BV-formalism (\cite{BV2} and \cite{Sw,Se}).

\begin{enumerate}
\item Consider the case $k=0$. This implies that $S$ is of degree $1$. Since the degree of the graded
symplectic form is zero, it induces an isomorphism between the coordinates of positive
and negative degree. Assuming that $S$ is non-trivial, there must be coordinates of positive degree -- and hence of negative degree as well.
We remark that the situation just described appears in the BFV-formalism (\cite{BF,BV1}).
\item Suppose $k>0$ and that all the coordinates are of non-negative degree. Dg symplectic manifolds
with that property were called dg symplectic $N$-manifolds by P. \v{S}evera, (letter nr. $8$ of \cite{SeLetters} and \cite{SeHomotopy}).
Let us look at the case $k=1$ and $k=2$ in more detail:
\begin{enumerate}
\item $k=1$. The graded symplectic structure induces an isomorphism between the coordinates of degree $0$, which we denote by $(x^{i})_{i=1}^{n}$, and the coordinates in degree $1$, which we denote by $(p_{i})_{i=1}^{n}$. All other degrees are excluded, since the would imply that coordinates of negative degree are around.

The Hamiltonian $S$ has degree $2$, so locally it must be of the form
\begin{align*}
S=\frac{1}{2}\sum_{i,j=1}^{n}\pi^{ij}(x)p_ip_j.
\end{align*}
Hence, locally $S$ corresponds to a bivector field and the classical master equation $\{S,S\}=0$
implies that $S$ actually corresponds to a Poisson bivector field.
This also holds globally, as the following Theorem due to A. Schwarz (\cite{Sw}) asserts:

\begin{theorem}\label{theorem:Schwarz}
Let $(\sM,\omega)$ be a graded symplectic structure of degree $1$. Then
$(\sM,\omega)$ is symplectomorphic to $T^{*}[1]M$, equipped with the standard symplectic
form. Moreover, one can choose $M$ to be an ordinary manifold.
\end{theorem}

The Theorem provides us with an isomorphism 
\begin{align*}
\smooth(\sM) \cong \smooth(T^{*}[1]M) = \Gamma(\wedge TM),
\end{align*}
which maps the bracket $\{-,-\}$ induced by $\omega$ to the Schouten-Nijenhuis bracket on $\Gamma(\wedge TM)$.
A smooth function $S$ of degree $2$ is mapped to a bivector field $\pi$ and if $S$ satisfies the classical master equation, $\pi$ will be Poisson.

Hence, there is a one-to-one correspondence
\begin{align*}
\xymatrix{
\textrm{isomorphism classes of dg symplectic $N$-manifolds of degree $1$}\ar@{<->}[d]^{1}_{1}\\
\textrm{isomorphism classes of Poisson manifolds}.
}
\end{align*}

\item $k=2$. It was noticed by P. \v{S}evera (see letter nr. $7$ of \cite{SeLetters}) that there is a one-to-one correspondence
\begin{align*}
\xymatrix{
\textrm{isomorphism classes of dg symplectic $N$-manifolds of degree $2$}\ar@{<->}[d]^{1}_{1}\\
\textrm{isomorphism classes of Courant algebroids}.
}
\end{align*}
We will not spell out the details of this correspondence -- the interested reader can find them in \cite{SeLetters} or \cite{R1}.

Observe that the degree of the graded symplectic form $\omega$ allows for coordinates
in degree $0$, $1$ and $2$. We denote them by $(x^{i})_{i=1}^{n}$, $(\xi^{\alpha})_{\alpha=1}^{A}$
and $(p_i)_{i=1}^{n}$, respectively. The graded symplectic form can be written as
\begin{align*}
\omega = \sum_{i=1}^{n}dp_i dx^{i} + \frac{1}{2}\sum_{\alpha,\beta=1}^{n}d\left(g_{\alpha\beta}(x)\xi^{\alpha}\right)d\xi^{\beta},
\end{align*}
where $(g_{\alpha\beta})$ is a symmetric non-degenerate form.

Globally, the graded symplectic form $\omega$ corresponds to $T^{*}[2]M$ and an additional
vector bundle $E$ over $M$, equipped with a non-degenerate fibre pairing $g$.

A Hamiltonian function $S$ for a cohomological vector field on such a graded manifold is locally
of the form
\begin{align*}
S= \sum_{i,\alpha}\rho^{i}_{\alpha}(x)p_i\xi^{\alpha} + \frac{1}{6}\sum_{\alpha,\beta,\gamma}f_{\alpha\beta\gamma}(x)\xi^{\alpha}\xi^{\beta}\xi^{\gamma}.
\end{align*}
The first term corresponds to a bundle map $\rho: E \to TM$, while the second one gives a bracket
$[-,-]$ on $\Gamma(E)$.

The classical master equation $\{S,S\}=0$ is equivalent to the statement that $(\rho,[-,-])$ equips
$(E,g)$ with the structure of a Courant algebroid.

\item Following J. Grabowski (\cite{G}), it is possible to include {\em generalized complex structures} on Courant algebroids in the graded
picture as follows: Assume that $S$ is a Hamiltonian function for a symplectic cohomological vector field
on a graded symplectic manifold of degree $2$. We want to construct a two-parameter family of
solutions of the classical master equation, i.e. we look for a smooth function $T$ of degree $3$ such that
\begin{align*}
\{\alpha S + \beta T, \alpha S + \beta T\}=0
\end{align*}
is satisfied for all constants $\alpha$ and $\beta$.

An important class of solutions arises when one finds a smooth function $J$ of degree $2$ such that $T=\{S,J\}$ satisfies
$\{T,T\}=0$. Under these circumstances, $T$ will solve the above two-parameter version of the classical master equation.
One way to assure $\{T,T\}=0$ is to require
\begin{align*}
\{\{S,J\},J\}=\lambda S
\end{align*}
to hold for some constant $\lambda$.
Up to rescaling, $\lambda$ is $-1$, $0$ or $1$. For $\lambda = -1$, such a smooth function $J$
yields a generalized complex structure on the Courant algebroid corresponding to $S$, if $J$ does not depend on the dual coordinates
$(p_i)_{i=1}^{n}$.
\end{enumerate}
\end{enumerate}

\subsection{Graded symplectic reduction}\label{subsection:reduction}

We recall some facts about reduction of presymplectic submanifolds. This can be extended to
graded symplectic manifolds and provides a unified approach to the reduction
of Poisson structures, Courant algebroids and generalized complex structures.
This Subsection relies on \cite{BCMZ1}, \cite{BCMZ2}, \cite{CZ1} and \cite{CZ2}, respectively.

\begin{definition}
Let $(M,\omega)$ be a symplectic manifold. A submanifold $i:S \hookrightarrow M$ is {\em presymplectic}
if the two-form $i^*\omega$ has constant rank.

In this case the kernel of $i^*\omega$ forms an integrable distribution $\mathcal{D}$ of $S$, called
the characterisitic distribution, and we denote its space of leaves
by $\underline{S}$.
\end{definition}

\begin{example}
A special class of presymplectic manifolds are {\em coisotropic} submanifolds.
A submanifold $S$ is coisotropic if for every point $x\in S$, the symplectic orthogonal of $T_xS$ is contained in $T_xS$. Equivalently, one can require that $S$ is given locally by the zero-set of constraints in
involution, i.e. there is a submanifold chart of $S$ with coordinates $(x^{i},y^{a})$ such that
\begin{itemize}
\item $S$ is given locally by $\{y^{a}\equiv 0\}$ and
\item the Poisson bracket $\{y^{a},y^{b}\}$ of any two constraints lies in the ideal of the algebra of smooth functions generated by the transverse coordinates $(y^{a})$.
\end{itemize}
\end{example}

\begin{lemma}
Let $i:S\hookrightarrow M$ be a presymplectic submanifold of $(M,\omega)$.
If $\underline{S}$ carries a smooth structure such that the canonical projection
\begin{align*}
\pi: S \to \underline{S}
\end{align*}
is a submersion, there is a unique symplectic structure $\underline{\omega}$
on $\underline{S}$ which statisfies
\begin{align*}
\pi^{*}\underline{\omega}=i^{*}\omega.
\end{align*}
\end{lemma}

\begin{remark}
From now on, we will always assume that $\underline{S}$ can be equipped with a smooth structure
such that the natural projection from $S$ is a submersion.
Assuming this, the symplectic manifold  $(\underline{S},\underline{\omega})$ is called the {\em reduction} of  $S$.
\end{remark}

\begin{definition}
Let $i:S\hookrightarrow M$ be a presymplectic submanifold of $(M,\omega)$.

A function $f \in \smooth(M)$ is called
\begin{itemize}
\item $S$-reducible, if $i^{*}f$ is invariant under the characteristic
distribution $\mathcal{D}$.
\item strongly $S$-reducible, if the Hamiltonian vector field of $f$ is tangent to $S$.
\end{itemize}
\end{definition}

\begin{lemma}
\hspace{0cm}
\begin{itemize}
\item Let $f$ be an $S$-reducible function. There is a unique smooth function $\underline{f}$ on $\underline{S}$
satisfying
\begin{align*}
\pi^{*}\underline{f}=i^{*}f.
\end{align*}

\item Let $f$ be strongly $S$-reducible. Then the following assertions hold:
\begin{enumerate}
\item $f$ is $S$-reducible.
\item The restriction $X_f|_{S}$ of the Hamiltonian vector field of $f$ to $S$ is projectable, i.e.
\begin{align*}
[X_f|_{S},\mathcal{D}] \subset \mathcal{D}.
\end{align*}
This implies that there is a unique vector field $\underline{X}_f$ on $\underline{S}$ which is $\pi$-related to $X_f|_{S}$, i.e. $T_{x}\pi(X_f)_x = (\underline{X}_f)_{\pi(x)}$ holds for all $x \in S$.
\item The vector field $\underline{X}_f$ is the Hamiltonian vector field of $\underline{f}$.
\end{enumerate}
\end{itemize}
\end{lemma}

\begin{proof}
Let $f$ be an $S$-reducible function. To establish the existence and uniqueness of $\underline{f}$, one observes that $i^{*}f$ is constant along the fibres of
\begin{align*}
\pi: S \to \underline{S}
\end{align*}
and hence descends to a function $\underline{f}$ on $\underline{S}$. Smoothness of $\underline{f}$ follows from the assumption that $\pi$ is a surjective submersion, which implies that
a function on $\underline{S}$ is smooth if and only if its pull back by $\pi$ is.

From now on, let $f$ be a strongly $S$-reducible function. We first show that $f$ is also $S$-reducible, i.e. that it is invariant under
the characteristic distribution $\mathcal{D}$ of $S$. The identities
\begin{align*}
L_X i^{*}f = \iota_X (d i^{*}f) \qquad \textrm{and} \qquad (i^{*}df)(X) = i^{*}\omega(X_f|_{S},X),
\end{align*}
where $X_f$ is the Hamiltonian vector field of $f$, imply
\begin{align*}
L_X i^{*}f = 0
\end{align*}
for all $X \in \Gamma(\mathcal{D})$, since $\mathcal{D}$ is the kernel of $i^{*}\omega$.

Next, we claim that the commutator of $X_f|_S$ and an arbitrary vector field $X$ on $S$, with values in $\mathcal{D}$, is a vector field
with values in $\mathcal{D}$ too. To this end, we compute
\begin{eqnarray*}
\iota_{[X,X_f|_{S}]}i^* \omega = ([L_X, \iota_{X_f|_S}]) i^{*}\omega = L_X (\iota_{X_f|_S}i^{*}\omega) = L_X i^{*} df = 0.
\end{eqnarray*}
This guarantees that
\begin{align*}
X_f|_S(\pi^{*}g)
\end{align*}
is $S$-reducible and one can define a linear endomorphism of $\smooth(\underline{S})$ by
\begin{align*}
g \mapsto \underline{X_f|_S(\pi^{*}g)}.
\end{align*}
It is easy to check that this is a derivation of $\smooth(\underline{S})$ and hence corresponds to a vector field $\underline{X}_f$ on $\underline{S}$.
By construction, $X_f|_S$ and $\underline{X}_f$ are $\pi$-related; uniqueness of $\underline{X}_f$ follows from $\pi$ being a surjective submersion.

The final claim follows from
\begin{align*}
\pi^{*}(d\underline{f}) = i^{*} df = \iota_{X_f|_S}i^*\omega = \pi^{*}( \iota_{\underline{X}_f}\underline{\omega}).
\end{align*}
\end{proof}

\begin{lemma}
Let $S$ be a coisotropic submanifold. Then the notion of $S$-reducibility is equivalent to
the notion of strong $S$-reducibility.
\end{lemma}

\begin{remark}
\hspace{0cm}
\begin{itemize}
\item The above definitions and statements can be extended to graded symplectic manifolds.
\item In the last Subsection, we saw that any Poisson manifold gives rise to a dg symplectic
manifold $(\mathcal{M},\omega)$ of degree $1$ with Hamiltonian function $\Theta$ of degree $2$.

Let $\mathcal{S}$ be a graded presymplectic submanifold of $\mathcal{M}$ and suppose that
$\Theta$ is strongly reducible. Graded reduction yields a new graded symplectic manifold $(\underline{\mathcal{S}},\underline{\omega})$ of degree $1$. Moreover, $\Theta$
induces a smooth function $\underline{\Theta}$ on $\underline{S}$ which satisfies
\begin{align*}
\{\underline{\Theta},\underline{\Theta}\} = X_{\underline{\Theta}}(\underline{\Theta}) = 
\underline{X}_{\Theta}(\underline{\Theta})= \underline{X_{\Theta}(\Theta)}=0.
\end{align*}
Hence, one obtains a new dg symplectic manifold of degree $1$, which corresponds to a new Poisson manifold.
\item Similarly, one can reduce dg manifolds of degree $2$, which correspond to
Courant algebroids. Furthermore, it is possible to include generalized complex
structures if one assumes that the corresponding function $J$ of degree $2$ is strongly reducible as well.
\end{itemize}
\end{remark}

\begin{remark}
It is also possible to extend Marsden-Weinstein reduction to graded symplectic manifolds by allowing Hamiltonian graded group actions, see \cite{CZ1}.
\end{remark}

\begin{example}
Let us consider the case of the dg symplectic manifolds arising from a Poisson manifold $(M,\pi)$ in more detail.
Recall that the corresponding graded symplectic manifold is $\mathcal{M}:=(T^{*}[1]M)$, equipped with the standard symplectic form.

Every coisotropic submanifold $\mathcal{S}$ corresponds to a submanifold $C$ of $M$ and an integrable distrubtion $B$ on $C$.
Reducibility of the Hamiltonian function $\Theta$, which corresponds to the Poisson bivector field $\pi$, is equivalent
to:
\begin{enumerate}
\item The image of the restriction of the bundle map $\pi^{\sharp}:T^{*}M \to TM$ to the conormal bundle $N^*C$ of $C$ is contained in $B$.
\item For any two function $f$ and $g$ on $M$ whose restriction to $C$ is invariant with respect to the distribution $B$, the restriction
of the Poisson bracket of $f$ and $g$ to $C$ is invariant too. 
\end{enumerate} 
Suppose that these conditions hold and that the leaf space $\underline{M}$ of $B$ admits a smooth structure such that the natural projection $p: C \to \underline{M}$ is
a surjective submersion.
Setting
\begin{align*}
\{\underline{f},\underline{g}\}:=\{\widetilde{p^{*}\underline{f}},\widetilde{p^{*}\underline{g}}\},
\end{align*}
defines a Poisson structure on $\underline{M}$. Here, $\widetilde{p^{*}\underline{f}}$ denotes a smooth extension of $p^{*}\underline{f}$ to $M$, 

What one recovers here is a particular case of Marsden-Ratiu reduction, see \cite{MR}.
More interesting examples can be obtained by considering presymplectic submanifolds, see \cite{CZ1,CZ2}.
\end{example}

\section{An introduction to the AKSZ-formalism}\label{section:AKSZ}

The AKSZ-formalism goes back to the article \cite{AKSZ} of M. Alexandrov, M. Kontsevich, A. Schwarz and O. Zaboronsky.
 It is a procedure that allows one to construct solutions to the classical master equation on mapping spaces between graded manifolds that are equipped
with additional structures. Particularly interesting examples arise from mapping spaces between shifted tangent bundles and
dg symplectic manifolds. This allows one to associate topological field theories to dg symplectic manifolds, encompassing examples such
as Chern-Simons theory (on trivial principal bundles) and the Poisson-sigma model. 

Let us describe the AKSZ-formalism in a nutshell. The `input data' are:
\begin{itemize}
\item The {\em source} $\mathcal{N}$: a dg manifold, equipped with a measure which is invariant under the cohomological vector field. 
\item The {\em target} $\mathcal{M}$: a dg symplectic manifold, whose cohomological vector field is Hamiltonian.
\end{itemize}

Out of this, the AKSZ-formalism tailors:
\begin{itemize}
\item a graded symplectic structure on the space of maps $\smooth(\mathcal{N},\mathcal{M})$ and
\item a symplectic cohomological vector field on $\smooth(\mathcal{N},\mathcal{M})$.
\end{itemize}

Under mild conditions, symplectic cohomological vector fields are Hamiltonian and one might find
a Hamiltonian function which commutes with itself under the Poisson bracket. Hence, in many cases the AKSZ-formalism produces
a solution of the classical master equation on $\smooth(\mathcal{N},\mathcal{M})$.

In the following, we $1)$ describe the relevant mapping spaces between graded manifolds and $2)$ give an outline of the AKSZ-formalism. 

\subsection{Maps of graded manifolds}

Given two graded manifold $X$ and $Y$, the set of morphisms $\Mor(X,Y)$ was defined to be the set of morphisms of $\mathbb{Z}$-graded algebras
from $\smooth(Y)$ to $\smooth(X)$.

The category of graded manifolds admits a monoidal structure $\times$, which is the coproduct of locally ringed spaces. From a categorical perspective,
one might wonder whether the set $\Mor(X,Y)$ can be equipped with the structure of a graded manifold in a natural way, such that it is the adjoint to the monoidal structure, i.e.
such that there is a natural isomorphism
\begin{align*}
\Mor(Z\times X, Y) \cong \Mor(Z,\Mor(X,Y))
\end{align*}
for arbitrary graded manifolds $X$, $Y$ and $Z$.

This turns out not to be possible. However, there actually is a (usually infinite dimensional) graded manifold $\Map(X,Y)$ canonically associated to a pair $(X,Y)$, which satisfies
\begin{align*}
\Mor(Z\times X, Y) \cong \Mor(Z,\Map(X,Y)).
\end{align*}
Moreover, there is a natural inclusion of $\Mor(X,Y)$ into $\Map(X,Y)$ as the submanifold of degree $0$.

\begin{remark}
\hspace{0cm}
\begin{enumerate}
\item Usually, $\Map(X,Y)$ is an infinite-dimensional object. However, there are noteworthy finite dimensional examples such as $\Map(\mathbb{R}[1],X)=T[1]X$.
\item The difference between $\Mor$ and $\Map$ can be illustrated in the following example: While $\Mor(X,\mathbb{R})$ is equal to the elements of $\smooth(X)$ in degree $0$,
the mapping space $\Map(X,\mathbb{R})$ is equal to the whole of $\smooth(X)$.
\item Similarly, the infinitesimal object associated to the group of invertible morphisms from $X$ to $X$ is the Lie algebra of vector fields of degree $0$,
whereas considering the group of invertible elements of the mapping space $\Map(X,X)$ yields the graded Lie algebra of all vector fields on $X$
\end{enumerate}
\end{remark}

\begin{remark}
Let us spell out $\Mor$ and $\Map$ in the local picture, i.e. for two graded vector spaces
$V$ and $W$. One has
\begin{align*}
\Mor(V,W) = \Mor(\smooth(W),\smooth(V))= (W \otimes \smooth(V))^{0}
\end{align*}
where $(W\otimes \smooth(V))$ is considered as a graded vector space and the superscript $0$ refers
to the elements in degree $0$.
In contrast to this,
\begin{align*}
\Map(V,W) = W\otimes \smooth(V)
\end{align*}
holds
\end{remark}

\subsection{Lifting geometric structures}

Geometric structures on the graded manifolds $X$ and $Y$ induce interesting
structures on the mapping space $\Map(X,Y)$. For instance, cohomogical vector fields on
$X$ and $Y$ can be lifted to commuting cohomological vector fields on $\Map(X,Y)$.
Another example is a graded symplectic structure on $Y$ and an invariant measure on $X$,
which allow one to construct a graded symplectic structure on $\Map(X,Y)$.
Let us elaborate on this in more detail:

\begin{enumerate}
\item The groups of invertible maps $\Diff(X)$ and $\Diff(Y)$ act on $\Map(X,Y)$ by composition and these
two actions commute. Differentiation yields commuting infinitesimal actions
\begin{align*}
\xymatrix{
\mathcal{X}(X)  \ar[rrr]^{L} &&& \mathcal{X}(\Map(X,Y)) &&& \mathcal{X}(Y) \ar[lll]_{R}.
}
\end{align*}
Now, suppose $Q_X$ and $Q_Y$ are cohomological vector fields on $X$ and $Y$.
We denote their images under $L$ -- respectively $R$ -- by $Q_X^L$ and $Q_Y^{R}$.
Their sum
\begin{align*}
Q:= Q_X^{L} + Q_Y^{R}
\end{align*}
is then a cohomological vector field on $\Map(X,Y)$.
\item Suppose $Y$ carries a graded symplectic form $\omega$. Any form $\alpha \in \Omega(Y)$
can be pulled back to a form on $\Map(X,Y)\times Y$ via the evaluation map
\begin{align*}
\textrm{ev}: \Map(X,Y) \times X \to Y.
\end{align*}

To produce a differential form on the mapping space $\Map(X,Y)$, some notion of
push forward along $X$ is required. To this end, the theory of Berezinian measures and Berezinian
integration is needed -- the interested reader might consult \cite{M}. Basically, this is an extension of the usual (Lebesgue-) integration theory
by adding the rule
\begin{align*}
\int \xi d\xi = 1
\end{align*}
for each odd coordinate $\xi$.
For instance, if one considers integration on a graded vector space $V$ concentrated in odd degrees, the integration map
\begin{align*}
\smooth(V) = \wedge V^{*} \to \mathbb{R}
\end{align*}
is just the projection to the top exterior product (which we identify with $\mathbb{R}$).
Another special case is provided by $X=T[1]\Sigma$, where $\Sigma$ is some compact, smooth, oriented manifold.
The graded manifold $X$ carries a canonical measure which maps a function $f$ on $X$ to
\begin{align*}
\int_{X} f := \int_{\Sigma}j(f),
\end{align*}
where $j$ denotes the isomorphism between $\smooth(X)$ and $\Omega(\Sigma)$.

Assuming that $X=T[1]\Sigma$ is equipped with its canonical measure, we obtain a map
\begin{eqnarray*}
\Omega(Y) &\to & \Omega(\Map(X,Y))\\
\alpha &\mapsto& \hat{\alpha}:=\int_{X}\textrm{ev}^{*}\alpha
\end{eqnarray*}
which is of degree $-\dim(\Sigma)$.

Furthermore, one can check that if $\omega$ is a graded symplectic form of degree $k$ on $Y$,
$\hat{\omega}$ is a graded symplectic form of degree $k-\dim(Y)$ on $\Map(X,Y)$.

\item We want to combine the structures obtained in $1.)$ and $2.)$. First, we assume that the cohomological vector field $Q_Y$
is Hamiltonian with Hamiltonian function $\Theta$. This implies that $Q_{Y}^{R}$ is also Hamiltonian
with respect to the graded symplectic structure $\hat{\omega}$ -- 
in fact, $\hat{\Theta}$ is a Hamiltonian function.

Concerning the source $X$, we concentrate on the case $X=T[1]\Sigma$, equipped with the
cohomological vector field $Q_X$ which corresponds to the de Rham differential.
The canonical measure on $X$ is invariant with respect to the cohomological vector field in the sense
that
\begin{align*}
\int_{X}Q_X(f) = 0
\end{align*}
holds for every smooth function $f$.
It follows from invariance that $Q_X^{L}$ is also Hamiltonian with respect to $\hat{\omega}$. We denote the Hamiltonian function
by $S_0$. 
\end{enumerate}

\begin{example}
Let $Y$ be a graded symplectic vector space of degree $1$, equipped with the standard exact symplectic form $d\alpha$. By Theorem \ref{theorem:Schwarz},
$Y$
is symplectomorphic to $T^{*}[1]V$, for some vector space $V$. Moreover, a symplectic cohomological vector field on $Y$ corresponds to a 
Poisson bivector field $\pi$ on $V$.
Suppose $(q^{i})$ are local coordinates on $V$ and denote the dual fibre coordinates on $T^{*}[1]V$ by $(p_i)$.
The set of coordinates $(q^{i},p_i)$ of $Y$, together with a set of coordinates $(x^{i}, dx^{i})$ of $X=T[1]\Sigma$, induces a set of coordinates $(X^i,\eta_i)$ on $\Map(X,Y)$.
In these coordinates, the Hamiltonian $S_0$ for the lift of the de Rham differential on $X$  to the mapping space reads
\begin{align*}
S_0 = \pm \int_{\Sigma}\eta_i dX^{i},
\end{align*}
while the Hamiltonian function for the lift of the cohomological vector field on $Y$ is given by
\begin{align*}
\hat{\Theta} = \frac{1}{2}\int_{\Sigma}\pi^{ij}(X)\eta_i\eta_j.
\end{align*} 
The sum of $S_0$ and $\hat{\Theta}$ is the BV action functional of the Poisson sigma model on $\Sigma$. This is a topological field theory associated to Poisson structures, which was discovered by N. Ikeda (\cite{I}) and P. Schaller and T. Strobl (\cite{SS}).

\end{example}

\begin{remark}
Assuming that the graded symplectic form $\omega$ on $Y$ is of degree $k$,
the Hamiltonian function $\Theta$ has degree $k+1$. If one assumes in addition that $\Sigma$ is of dimension $n$,
the graded symplectic form $\hat{\omega}$ on $\Map(X,Y)$ is of degree $k-n$ and the Hamiltonian function $\hat{\Theta}$ is of degree $k+1-n$.

The case $n=k+1$ corresponds to the BV-formalism -- introduced by Batalin and Vilkovisky (\cite{BV2}), while
$n=k$ corresponds to the BFV-formalism -- introduced by Batalin, Fradkin and Vilkovisky (\cite{BF,BV1}). These together with the cases $k<n$ should be related to extended
topological fields theories in the sense of Lurie (\cite{L}).
\end{remark}

\end{document}